\documentclass[10pt]{article}

\usepackage[margin=1in]{geometry}
\usepackage{amsmath,amssymb,amsthm,bm}
\usepackage{booktabs}
\usepackage{graphicx}
\usepackage{float}
\usepackage{placeins}
\usepackage{microtype}
\usepackage[numbers,sort&compress]{natbib}
\usepackage[colorlinks=true,allcolors=blue]{hyperref}
\usepackage[nameinlink,noabbrev]{cleveref}

\newtheorem{theorem}{Theorem}
\newtheorem{lemma}{Lemma}

\newcommand{\Rerr}{R_{\mathrm{err}}}

\newcommand{\diag}{\operatorname{diag}}

\title{Reliability--Contagion Feasibility in LLM Multi-Agent Networks}
\author{
Ruiwu Niu$^{1}$, Xincheng Shu$^{2,3}$, and Ying Zhao$^{4}$\\[0.6em]
\parbox{0.94\textwidth}{\centering\small
$^{1}$Department of Applied Data Science, Hong Kong Shue Yan University, Hong Kong SAR, China\\
$^{2}$Computational Communication Research Center, Beijing Normal University, Zhuhai 519087, China\\
$^{3}$School of Journalism and Communication, Beijing Normal University, Beijing 100875, China\\
$^{4}$Department of Electrical Engineering, City University of Hong Kong, Kowloon, Hong Kong SAR, China\\
Emails: \texttt{rniu@hksyu.edu} (R. Niu); \texttt{yzhao396-c@my.cityu.edu.hk} (Y. Zhao)
}}
\date{Draft of July 23, 2026}

\begin{document}
\maketitle

\begin{abstract}
Communication allows large language model agents to pool evidence, but it also creates paths along which an erroneous claim can spread. We formulate a correction-aware network model that tracks susceptible, exposed, infectious, and corrected agents and derive its early-invasion condition for heterogeneous communication networks. We then couple this propagation model to an analytic majority-vote benchmark in which a clean-task reliability target imposes a minimum connectivity requirement. Under fixed exposure per communication edge, reliability and error control impose opposing graph constraints. We characterize when their intersection is empty and when it contains an intermediate connectivity range, and identify regular graphs that attain the smallest invasion factor within the reliable graph class when such graphs exist. Under a fixed sender budget, the homogeneous first-order threshold is independent of network density, showing that the communication-budget convention determines whether added edges increase early propagation risk. Finite-network simulations on 21,000 trajectories illustrate these directional predictions. A controlled \texttt{grok-4.3} experiment then evaluates three six-node topologies on 36 new closed-world tasks, with a balanced 12-task subset continued to full cascades. Mean first-generation offspring increased from 0.667 to 1.333 and 1.667 as degree increased from 2 to 4 and 5, while the adoption fraction among exposed neighbours remained 0.333. Mean non-seed erroneous adoption in the full-cascade subset was 0.200, 0.333, and 0.333. Together, these results provide a tractable basis for selecting connectivity under explicit reliability and propagation constraints.
\end{abstract}

\section{Introduction}

Large language models are increasingly organized as multi-agent systems in which several model instances propose answers, exchange intermediate reasoning, critique one another, and aggregate a final decision. Communication can improve reasoning and factuality by exposing an agent to evidence or solution paths that it did not produce independently \citep{du2024multiagent,liang2024divergent}. The gain depends on the interaction protocol. Multi-agent debate does not consistently outperform voting, ensembling, or self-consistency across tasks and hyperparameter choices \citep{smit2024mad,choi2025debate}. These findings make the communication graph part of the model design.

Communication also couples agent failures. An agent can conform to an incorrect peer or retain an unsupported premise across later rounds \citep{baltaji2024conformity}. Adversarial work shows related propagation mechanisms. Agent Smith demonstrated infectious jailbreak behavior under randomized pairwise interaction among multimodal agents \citep{gu2024agentsmith}, and communication attacks can alter messages in established multi-agent workflows \citep{he2025communication}. These settings differ from ordinary factual errors, yet they establish that an interaction channel can amplify a failure after it enters the network.

Topology has consequently become an explicit design variable. Sparse debate graphs can reduce token cost while retaining task accuracy in the evaluated settings \citep{li2024sparse}. GPTSwarm represents language-agent systems as optimizable computational graphs \citep{zhuge2024gptswarm}; AgentPrune removes redundant or harmful messages \citep{zhang2025agentprune}; and G-Designer learns task-dependent communication graphs \citep{zhang2025gdesigner}. This literature primarily optimizes task accuracy, robustness, or cost for a given evaluation distribution.

The closest empirical antecedents study error propagation directly. Shen et al. \citep{shen2025propagation} use counterfactual interventions to measure how an injected correct or incorrect agent output changes the final system decision across communication topologies. They find that moderately sparse graphs can suppress errors while preserving useful information and introduce a learned topology design. Their contribution is empirical causal analysis and topology optimization; it does not derive an epidemic threshold or a graph-class feasibility boundary. NetSafe reports that topology and connectivity also affect vulnerability under several multi-agent attacks \citep{yu2025netsafe}. These studies establish empirical links among sparsity, information flow, and robustness. The remaining question is analytic: at fixed node count, when does a graph exist that meets both a clean-reliability requirement and a subcritical error-propagation requirement?

Network epidemiology provides a natural language for the propagation side. Next-generation operators define reproduction ratios in heterogeneous compartmental systems \citep{diekmann1990r0,vandendriessche2002reproduction}, and node-level network mean-field models relate invasion conditions to the spectral radius of an interaction matrix \citep{vanmieghem2009virus}. Our earlier SEIHR studies provide part of the methodological background for the present work. A deterministic SEIHR model treated exposed individuals as infectious, represented hospitalization explicitly, incorporated human migration, and derived a basic reproduction number for studying control measures \citep{niu2020seihr}. A stochastic extension described daily fluctuations in infections and hospitalizations and established sufficient conditions for stochastic stability of the disease-free equilibrium \citep{niu2021stochasticseihr}. The present SEICS formulation retains interpretable compartments and threshold analysis while assigning states to observable claim-handling events on an explicit agent communication network.

Rumor models have also adapted epidemic tools to information diffusion on heterogeneous networks \citep{nekovee2007rumour}. They include memory loss, denial, and correction. The SIHR model represents forgetting and remembering \citep{zhao2012sihr}. Liu et al. \citep{liu2019cyberrumor} give a population-level four-compartment model in which network structure is compressed into an effective contact parameter. Their rumor-neutral, rumor-received, rumor-believed, and rumor-denied populations include correction and forgetting, and the model yields a next-generation threshold. These results are mathematical precedents for a correction-aware claim process.

The LLM setting adds two operational requirements. First, a state must be assigned from message traces rather than from a latent belief. We therefore follow one pre-specified claim and record its reception, use, retransmission, and explicit correction. Second, topology changes communication volume unless the budget convention is fixed. With constant exposure per edge, adding links increases a sender's total influence. With a fixed sender budget, the same exposure is divided among more outgoing links. Contact-based epidemic models already show that reactive and degree-normalized contact processes have different spectral thresholds \citep{gomez2010contact}. This distinction is equally important in agent systems.

Recent preprints have proposed mathematical descriptions of multi-agent error propagation. Collective Hallucination studies dynamic network effects and defenses for unsupported claims \citep{jamshidi2026collective}. Contagion Networks models evaluator-preference transmission with an empirically estimated cross-agent contagion matrix and spectral propagation regimes \citep{liu2026contagion}. Its propagated variable and operator differ from the claim-state SEICS process considered here, and it does not impose a clean-reliability degree constraint.

This paper develops a compact theory around three results. First, a heterogeneous claim-level SEICS model yields a next-generation operator that separates exposure, adoption, pre-transmission rejection, reactive correction, and correction retention. Second, an independent-signal majority benchmark converts a clean-reliability target into a minimum-degree condition. Combining it with the spectral invasion condition gives graph-class existence and impossibility results and defines an explicit feasibility problem for LLM-agent communication. Third, a fixed-sender analysis shows exactly where the density effect disappears from the homogeneous spectral threshold. The normalized-threshold identity has a classical contact-process precedent; its role here is to expose the communication-budget assumption behind the LLM topology trade-off.

The empirical scope is deliberately small. We use finite-network Gillespie simulations to check implementation consistency and the direction of cascade risk under the stated process. A controlled real-model network experiment then measures first-generation propagation on 36 new tasks and full-cascade reach on a balanced 12-task subset along a nested edge-addition path. The analytic reliability benchmark and the propagation experiments remain separate pieces of evidence.

\section{Claim-level network SEICS model}
\label{sec:model}

Let \(G=(V,E)\) contain \(N\) agents. We use the convention \(A_{ij}=1\) when agent \(j\) can send to agent \(i\). The analysis follows one externally checkable false claim. Each node occupies one observable state:

\begin{itemize}
    \item \(S\): the claim is absent from the active trace and no retained correction is present;
    \item \(E\): the claim has been received, but the agent has not asserted or retransmitted it;
    \item \(I\): the agent asserts, uses, or retransmits the claim as true;
    \item \(C\): the agent explicitly rejects the claim while correction evidence remains available.
\end{itemize}

\Cref{fig:seics} summarizes the observable transition pipeline. Incoming messages carrying the claim move a susceptible node from S to E; adoption then moves it to I, and explicit correction moves it to C. Two correction routes bypass retransmission: an exposed node can reject the claim before entering I, and a susceptible node can verify the claim prophylactically. Loss of retained correction returns C to S.

\begin{figure}[H]
    \centering
    \includegraphics[width=\linewidth]{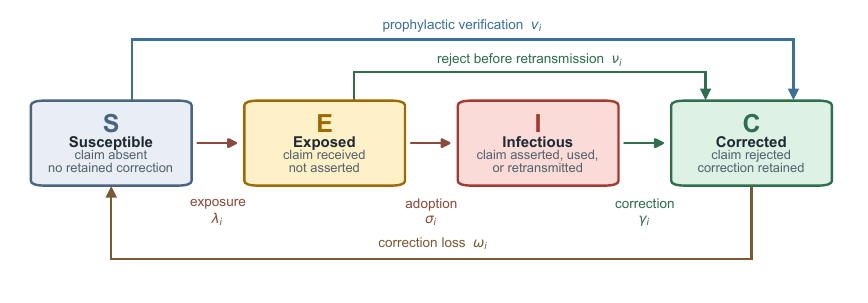}
    \caption{Claim-level SEICS transition pipeline. The main path is S$\rightarrow$E through incoming exposure at rate $\lambda_i=\sum_jT_{ij}x_j$, E$\rightarrow$I through adoption at rate $\sigma_i$, and I$\rightarrow$C through reactive correction at rate $\gamma_i$. An exposed node can reject the claim before retransmission through E$\rightarrow$C at rate $\nu_i$, and prophylactic verification moves S$\rightarrow$C at rate $v_i$. Correction loss returns C$\rightarrow$S at rate $\omega_i$. States are assigned from observable message traces. The diagram represents the deterministic node-level mean-field mechanism in \eqref{eq:seics}; it is not an empirical result.}
    \label{fig:seics}
\end{figure}

Let \(T_{ij}\ge0\) be the instantaneous exposure rate from \(j\) to \(i\). For node probabilities \(s_i,e_i,x_i,c_i\), define \(\lambda_i=\sum_jT_{ij}x_j\). The deterministic quenched mean-field model is
\begin{align}
\dot s_i&=-s_i\lambda_i-v_i s_i+\omega_i c_i,\nonumber\\
\dot e_i&=s_i\lambda_i-(\sigma_i+\nu_i)e_i,\nonumber\\
\dot x_i&=\sigma_i e_i-\gamma_i x_i,\label{eq:seics}\\
\dot c_i&=\nu_i e_i+\gamma_i x_i+v_i s_i-\omega_i c_i.\nonumber
\end{align}
Here \(\sigma_i\) is adoption, \(\nu_i\) is rejection before retransmission, \(\gamma_i\) is reactive correction, \(\omega_i\) is correction loss, and \(v_i\) is prophylactic verification. Equation \eqref{eq:seics} uses the closure \(\Pr(S_i,I_j)\approx s_ix_j\). The exact finite stochastic process has joint node states and is treated separately in \cref{sec:numerical}.

At the claim-free equilibrium,
\[
s_i^0=\frac{\omega_i}{v_i+\omega_i},\qquad
c_i^0=\frac{v_i}{v_i+\omega_i}.
\]
Define
\[
P=\diag\!\left(s_i^0\frac{\sigma_i}{\sigma_i+\nu_i}\right),
\qquad \Gamma=\diag(\gamma_i),
\qquad L=T\Gamma^{-1}.
\]
The matrix \(L\) records expected lifetime exposure during one infectious episode.

\begin{theorem}[Heterogeneous early-invasion condition]
\label{thm:ngm}
Assume \(T\ge0\), \(\sigma_i,\gamma_i>0\), \(\nu_i,v_i,\omega_i\ge0\), and \(v_i+\omega_i>0\). The claim-level next-generation matrix is
\begin{equation}
K=PL=PT\Gamma^{-1},\qquad \Rerr=\rho(K).
\label{eq:ngm}
\end{equation}
On the invariant product of node probability simplices, the claim-free equilibrium of \eqref{eq:seics} is locally asymptotically stable when \(\Rerr<1\) and unstable when \(\Rerr>1\), subject to the usual irreducibility condition.
\end{theorem}

\begin{proof}
Let \(S^0=\diag(s_i^0)\), \(\Sigma=\diag(\sigma_i)\), and \(N_\nu=\diag(\nu_i)\). The linearized infected subsystem \((e,x)\) has the standard decomposition
\[
F=\begin{pmatrix}0&S^0T\\0&0\end{pmatrix},\qquad
V=\begin{pmatrix}\Sigma+N_\nu&0\\-\Sigma&\Gamma\end{pmatrix}.
\]
The nonzero next-generation block of \(FV^{-1}\) is
\[
S^0T\Gamma^{-1}\Sigma(\Sigma+N_\nu)^{-1}.
\]
The matrices \(AB\) and \(BA\) have the same nonzero eigenvalues. Since the diagonal factors commute, the spectral radius equals that of
\[
\Sigma(\Sigma+N_\nu)^{-1}S^0T\Gamma^{-1}=PT\Gamma^{-1}.
\]
The local result follows from the next-generation criterion \citep{vandendriessche2002reproduction}. At equality the linearization is non-hyperbolic.
\end{proof}

Theorem \ref{thm:ngm} is exact for the deterministic system \eqref{eq:seics}. In a finite continuous-time Markov chain, extinction remains possible above one and finite cascade probabilities vary smoothly with parameters. We use \(\Rerr\) as a mean-field early-invasion factor.

\section{Clean reliability benchmark}
\label{sec:reliability}

This benchmark is a fully specified analytic test rather than a task dataset or a fitted model-accuracy curve. It holds the network size \(N\) fixed and varies a node's degree \(d\). Each node receives one private signal and one signal from each of its \(d\) neighbours; all \(m=d+1\) signals are independent, have the same probability \(p>1/2\) of being correct, and are combined once by majority vote. A strict majority is selected, and an even tie is broken fairly. The benchmark excludes repeated discussion, source weighting, cross-agent correlation, and attention limits. Its outputs are a degree-to-accuracy curve and the minimum degree needed to reach a target accuracy. In plain terms, it asks how many independent opinions a node needs before a simple majority reaches the required accuracy. Define
\begin{equation}
B_m(p)=\Pr\!\left[\operatorname{Bin}(m,p)>\frac m2\right]
+\frac12\Pr\!\left[\operatorname{Bin}(m,p)=\frac m2\right],
\qquad Q(d;p)=B_{d+1}(p).
\label{eq:reliability}
\end{equation}

\begin{lemma}[Monotonicity]
For \(p\ge1/2\), \(B_m(p)\) is non-decreasing in \(m\). In particular,
\[
B_{2r}(p)=B_{2r-1}(p),
\]
and
\[
B_{2r+1}(p)-B_{2r}(p)
=\left(p-\frac12\right)\binom{2r}{r}[p(1-p)]^r\ge0.
\]
\end{lemma}

\begin{proof}
Condition on the final Bernoulli signal and separate the central binomial event. Terms outside the central event cancel, yielding the two identities.
\end{proof}

For a target \(Q_{\min}\), let
\begin{equation}
k_\star=\min\{d\in\{0,\ldots,N-1\}:Q(d;p)\ge Q_{\min}\}.
\label{eq:kstar}
\end{equation}
When this set is nonempty, every node meets the reliability target exactly when the graph satisfies \(\delta(G)\ge k_\star\). Equation \eqref{eq:reliability} is an analytic surrogate. Correlated agent messages, shared model errors, and attention limits can change the empirical reliability curve. Figure~\ref{fig:frontier}a plots this benchmark for the illustrative parameter choice used below; it is separate from the real-model experiment.

\section{Reliable and subcritical graph classes}
\label{sec:feasibility}

Consider homogeneous effective susceptibility \(q=s^0\sigma/(\sigma+\nu)\) and lifetime exposure \(L=\tau A\), where \(\tau\ge0\) is fixed per edge. Then
\begin{equation}
K=q\tau A,\qquad \Rerr(G)=q\tau\rho(A).
\label{eq:peredge}
\end{equation}

\begin{theorem}[Reliability--contagion feasibility boundary]
\label{thm:feasibility}
Fix \(N\), \(p>1/2\), and a finite \(k_\star\in\{1,\ldots,N-1\}\) from \eqref{eq:kstar}. Let \(q,\tau\ge0\). Over connected, simple, unweighted, undirected graphs on \(N\) nodes, every graph meeting the per-node reliability target satisfies
\begin{equation}
\Rerr(G)\ge q\tau k_\star.
\label{eq:lowerbound}
\end{equation}
Consequently:
\begin{enumerate}
    \item if \(q\tau k_\star\ge1\), no graph in the reliable class is strictly subcritical;
    \item if \(q\tau k_\star<1\) and a connected simple \(k_\star\)-regular graph exists, that graph is reliable, strictly subcritical, and minimizes \(\Rerr\) over the reliable graph class.
\end{enumerate}
\end{theorem}

\begin{proof}
For an undirected graph, the Rayleigh quotient at the all-ones vector gives
\[
\rho(A)\ge \bar d\ge\delta(G)\ge k_\star.
\]
Substitution into \eqref{eq:peredge} proves \eqref{eq:lowerbound}. A \(k_\star\)-regular graph has \(\rho(A)=k_\star\), attains the lower bound, and is strictly subcritical under the second condition.
\end{proof}

In plain terms, accuracy requires enough links for agents to compare independent answers, while error control requires few enough effective links to keep a mistake from multiplying; a usable network exists only when both requirements can be met. This interpretation retains the graph-class, homogeneity, and fixed-per-edge assumptions stated in the theorem.

The constructive part is conditional on the usual degree and parity requirements for a connected regular graph. When \(q\tau>0\), the feasible degrees in a regular density sweep satisfy
\begin{equation}
k_\star\le k<\frac{1}{q\tau}.
\label{eq:window}
\end{equation}
This interval is the precise intermediate-connectivity region in the homogeneous benchmark. If \(q\tau=0\), the propagation constraint is vacuous and every feasible regular degree \(k\ge k_\star\) is subcritical.

For illustration, set \(N=32\), \(p=0.65\), \(Q_{\min}=0.90\), and \(\tau=0.12\). Equation \eqref{eq:kstar} gives \(k_\star=16\), and the critical susceptibility at the lower bound is \(q_{\mathrm{crit}}=1/(0.12\times16)=0.5208\). With \(q=1/6\), every simple 32-node graph is subcritical and reliability selects the graph. With \(q=1/3\), reliable and subcritical graphs exist, but they require \(\rho(A)<25\). With \(q=2/3\), every reliable graph is supercritical under the model. Figure~\ref{fig:frontier} assembles the calculation: panel a gives the clean-reliability curve, panel b shows how the per-edge invasion factor rises with regular degree, panel c displays the intersection of the two constraints, and panel d previews the fixed-sender contrast developed in \cref{sec:budget}. These values are illustrative inputs.

\begin{figure}[t]
    \centering
    \includegraphics[width=\linewidth]{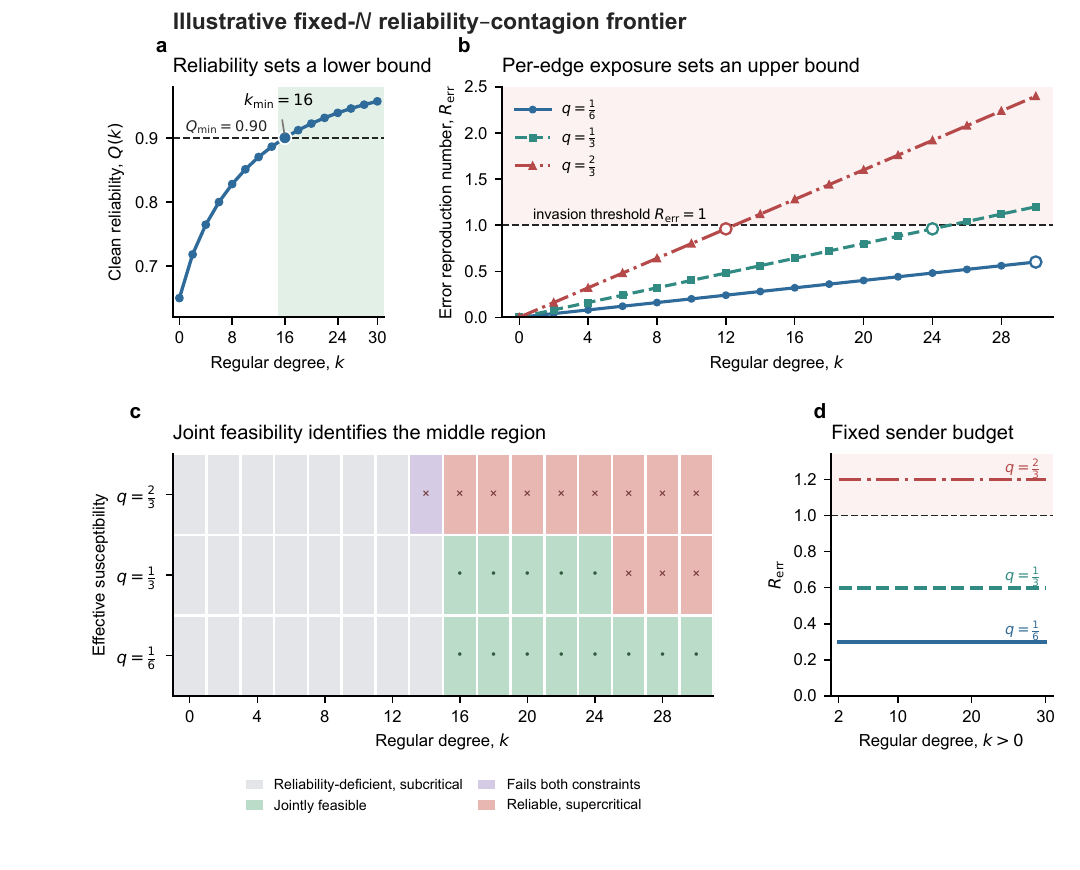}
    \caption{Illustrative fixed-\(N\) reliability--contagion calculation. (a) The independent-signal benchmark reaches \(Q_{\min}=0.90\) at \(k_\star=16\). (b) Under fixed per-edge lifetime exposure, \(\Rerr=q\tau k\) grows with degree. Open symbols mark the largest tested degree satisfying \(\Rerr<1\) for each susceptibility level. (c) Intersecting the analytic reliability requirement and strict subcriticality gives a feasible degree range for two susceptibility levels and no feasible degree for the highest level. (d) Under a fixed sender budget, the homogeneous spectral threshold is constant across positive regular degrees. The three effective-susceptibility values \(q\in\{1/6,1/3,2/3\}\) are illustrative inputs.}
    \label{fig:frontier}
\end{figure}
\FloatBarrier

\section{Fixed sender budget}
\label{sec:budget}

Let each infectious sender have total lifetime exposure budget \(b\), divided across its outgoing edges. With column \(j\) denoting sender \(j\), let
\[
D_{\mathrm{out}}=\diag(d_j^{\mathrm{out}}),\qquad
L=bAD_{\mathrm{out}}^{-1}.
\]

\begin{theorem}[Sender-normalized spectral threshold]
\label{thm:sender}
If \(b\ge0\), every node has positive out-degree, and susceptibility is homogeneous, then
\begin{equation}
\Rerr(G)=qb
\label{eq:sender}
\end{equation}
for every directed topology, including reducible and irregular graphs.
\end{theorem}

\begin{proof}
The matrix \(M=AD_{\mathrm{out}}^{-1}\) is column stochastic. Thus \(1\) is an eigenvalue and its induced one-norm is one, giving \(\rho(M)=1\). Therefore \(\rho(qbM)=qb\).
\end{proof}

In plain terms, if every sender splits the same total communication allowance among its recipients, adding recipients weakens each individual link, so network density alone does not change how easily a newly introduced error starts to multiply.

This result is the direct matrix analogue of the degree-normalized contact-process limit \citep{gomez2010contact}. Figure~\ref{fig:frontier}d shows the resulting degree-invariant line for the homogeneous example. The theorem concerns the linear spectral threshold. Seed reachability, finite cascade size, extinction time, and transient amplification can remain topology-dependent. A density effect in the linear spectral invasion factor under a fixed sender budget points to heterogeneity or another departure from the homogeneous normalized model; topology-dependent finite-cascade statistics do not by themselves contradict the theorem.

\section{Finite-network stochastic illustration}
\label{sec:numerical}

\subsection{Design}

We implemented the finite-state continuous-time SEICS process with an exact Gillespie algorithm on \(N=32\) circulant regular graphs. Degrees were \(k\in\{8,12,16,20,24,28,30\}\), susceptibilities were \(q\in\{1/6,1/3,2/3\}\), and communication used either lifetime exposure \(\tau=0.12\) per edge or lifetime sender budget \(b=1.8\). The rates were \(\gamma=1\) and \(\omega=0.1\); \(\sigma=q\) and \(\nu=1-q\), so \(\sigma/(\sigma+\nu)=q\). Each of the 42 cells contained 500 independently derived random-number streams, for 21,000 trajectories. A cascade was recorded when at least seven distinct nodes entered I. Wilson intervals summarize this binary endpoint.

The clean reliability curve was evaluated separately from \eqref{eq:reliability}. The Gillespie process does not generate the clean signals or their majority vote. The experiment therefore combines an analytic reliability constraint with a stochastic contagion illustration.

\subsection{Results}

Under fixed per-edge exposure, the analytic invasion factor increased with degree and outbreak probability increased overall. At \(q=1/3\), the reliability-qualified cells moved from \(\Rerr=0.64\) and outbreak probability 0.046 at \(k=16\) to \(\Rerr=1.20\) and probability 0.196 at \(k=30\). At \(q=2/3\), every reliability-qualified degree was analytically supercritical; outbreak probability rose from 0.266 at \(k=16\) to 0.584 at \(k=30\).

Under the fixed sender budget, the analytic factors were 0.3, 0.6, and 1.2 at every tested degree, as required by \cref{thm:sender}. Finite outbreak probabilities still varied with degree. For \(q=2/3\), they ranged from 0.158 to 0.282. This difference illustrates the gap between a linear threshold and a complete finite-cascade statistic. The maximum probability that a run still contained an E or I node at the horizon was 0.01; extending the horizon could therefore revise the reported ever-I outbreak probability by at most one percentage point in any cell.

Figure~\ref{fig:gillespie}a shows the degree trend under fixed per-edge exposure, including the separately computed reliability-qualified region. Figure~\ref{fig:gillespie}b shows the sender-budget case, where the analytic factor is degree-invariant but finite cascade probabilities are not. Figure~\ref{fig:gillespie}c places both regimes against the analytic invasion factor and marks the deterministic mean-field boundary.

\begin{figure}[t]
    \centering
    \includegraphics[width=\linewidth]{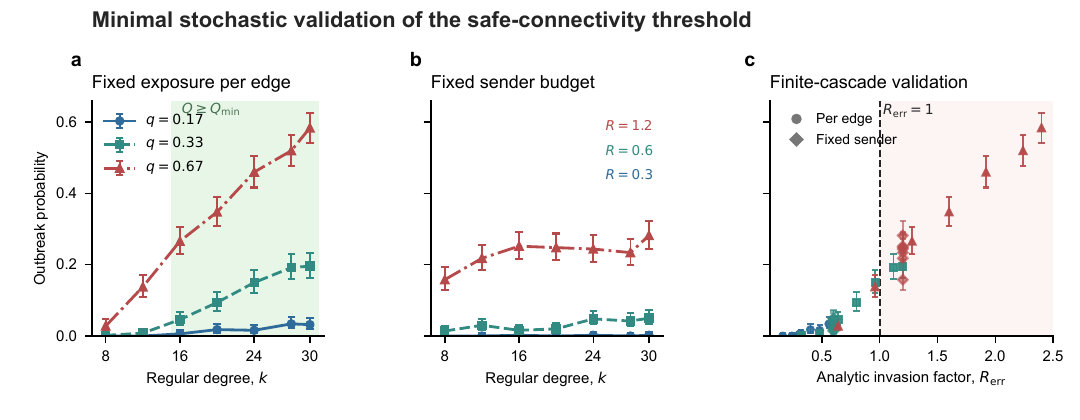}
    \caption{Finite-network stochastic illustration. Points show outbreak probabilities from 500 Gillespie trajectories per cell; error bars are two-sided 95\% Wilson intervals. (a) With fixed exposure per edge, outbreak probability increases overall with degree, particularly when the analytic factor crosses one. The green region marks degrees that meet the separate analytic reliability requirement. (b) With a fixed sender budget, the spectral factor remains constant within each susceptibility level while finite cascade probabilities retain degree dependence. (c) Outbreak probability plotted against the analytic invasion factor. The dashed line marks the deterministic mean-field boundary, rather than an exact finite-network transition.}
    \label{fig:gillespie}
\end{figure}
\FloatBarrier

\section{Controlled real-model network experiment}
\label{sec:real}

\subsection{Design}

We conducted a controlled real-model experiment to check whether added edges increase error spread under a frozen communication protocol. Six stateless instances of \texttt{grok-4.3} occupy three nested graphs: a cycle \(G_2\), a degree-four circulant graph \(G_4\), and the complete graph \(G_5\). The experiment uses 36 new fictional closed-world binary tasks, with 18 correct A answers and 18 correct B answers. None of the four pilot tasks used in protocol development was retained.

Node private-source reliabilities use the fixed multiset \((0.74,0.78,0.80,0.82,0.86,0.90)\), while each erroneous peer report has displayed reliability 0.80. The 36 tasks form a complete six-by-six crossing of seed node and cyclic reliability rotation. The six topology execution orders follow a Latin-square schedule, so each order occurs six times. The controller injects one exogenous false seed in every trial.

Propagation uses synchronous frozen generations. Each node is queried once, at first exposure, with all newly infectious neighbours listed separately. A false answer enters I, a correct answer enters C, and \texttt{UNSURE} is the observed E endpoint. Only I nodes transmit in the next generation. The primary endpoint is first-generation seed offspring \(Z_1\), measured for all 36 tasks under all three topologies. In plain terms, it counts how many immediate neighbours accept the seeded false claim. A pre-registered 12-task subset continues to full cascades; this subset contains each seed node, reliability rotation, and topology execution order twice. Its secondary endpoint is the fraction \(Y\) of the five non-seed nodes that ever enter I, which records how far the false claim eventually reaches.

The remaining 24 tasks stop after the first generation by design and are not treated as complete cascades. The protocol permits at most 444 logical decisions and 520 HTTP attempts, including retries. All comparisons are paired within task and reported descriptively.

\subsection{Results}

All 108 task-topology trial units completed. The realized cascades required 420 logical decisions and 422 HTTP attempts, including two successful retries and no failed decision record. All 420 response identifiers were present and unique. Recorded usage was 221,120 input tokens and 205,822 output tokens, for 426,942 tokens in total. The frozen-manifest, append-only-ledger, balance, cascade-scope, and budget checks passed.

The model followed the displayed source-reliability rule in all 420 decisions. It selected the erroneous peer report when reliability 0.80 exceeded private reliability 0.74 or 0.78, returned \texttt{UNSURE} at the 0.80 tie, and retained the correct private report at reliabilities 0.82, 0.86, and 0.90. This exact compliance confirms execution of the controlled mechanism; it does not estimate unconstrained hallucination susceptibility.

The topology summaries are shown in \cref{tab:real}. Mean seed offspring increased from 0.667 on \(G_2\) to 1.333 on \(G_4\) and 1.667 on \(G_5\). The corresponding mean fraction of exposed neighbours that entered I was 0.333 under every topology. The absolute increase therefore came from exposing more neighbours under the fixed per-edge protocol, rather than from a change in per-exposure adoption.

\begin{table}[t]
    \centering
    \small
    \caption{Controlled real-model outcomes. The primary endpoint \(Z_1\) uses all 36 tasks; the secondary non-seed ever-I fraction \(Y\) and infection depth use the pre-registered 12-task full-cascade subset. Values are arithmetic means over the indicated paired tasks.}
    \label{tab:real}
    \begin{tabular}{lrrrrrrr}
        \toprule
        Topology & \(k\) & \(n_{Z_1}\) & Mean \(Z_1\) & Adopted/exposed & \(n_Y\) & Mean \(Y\) & Mean depth \\
        \midrule
        Cycle & 2 & 36 & 0.667 & 0.333 & 12 & 0.200 & 1.000 \\
        Circulant & 4 & 36 & 1.333 & 0.333 & 12 & 0.333 & 1.333 \\
        Complete & 5 & 36 & 1.667 & 0.333 & 12 & 0.333 & 1.000 \\
        \bottomrule
    \end{tabular}
\end{table}

For the paired primary contrasts, \(G_4-G_2\) was positive for 24 tasks, zero for 12, and negative for none; \(G_5-G_4\) was positive for 12, zero for 24, and negative for none. The mean differences were 0.667 and 0.333, respectively. The direct \(G_5-G_2\) comparison had mean 1.000, with 24 positive, 12 zero, and no negative task-level differences.

In the 12-task full-cascade subset, mean non-seed ever-I fraction increased from 0.200 on \(G_2\) to 0.333 on \(G_4\) and remained 0.333 on \(G_5\). Four paired \(G_4-G_2\) values increased by 0.4 and eight were unchanged. All 12 \(G_5-G_4\) differences were zero. Cumulative reach therefore increased between the sparse and intermediate graphs and then saturated on the complete graph.

Figure~\ref{fig:real}a shows the nested edge-addition path. Panel b reports all 36 paired first-generation outcomes and their means; panel c counts positive, zero, and negative paired changes; panel d reports all 12 paired full-cascade outcomes. The points and lines are descriptive task-level values, not uncertainty intervals.

\begin{figure}[H]
    \centering
    \includegraphics[width=\linewidth]{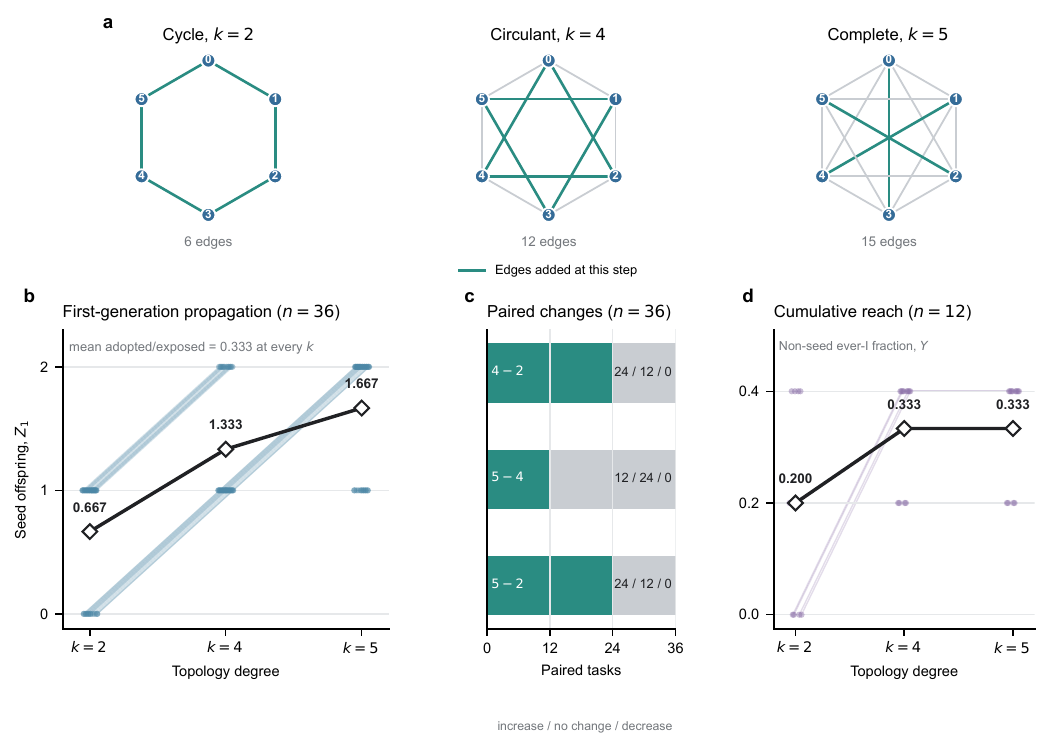}
    \caption{Controlled \texttt{grok-4.3} network experiment. (a) Nested edge-addition path on six nodes; highlighted edges are added at each step. (b) First-generation seed offspring for all \(n=36\) paired tasks. Thin lines connect the same task across degrees; black diamonds give arithmetic means. Mean adoption among exposed neighbours is 0.333 at every degree. (c) Counts of paired increases, no changes, and decreases for the degree \(4-2\), \(5-4\), and \(5-2\) contrasts. (d) Non-seed ever-I fraction for the pre-registered \(n=12\) full-cascade subset, with black diamonds showing arithmetic means. Raw markers use small horizontal offsets for visibility. No interval in this figure is an inferential confidence interval.}
    \label{fig:real}
\end{figure}
\FloatBarrier

\section{Discussion}

The feasibility theorem gives a direct answer to the sparsity question for fixed-size networks under stated assumptions. Connectivity below \(k_\star\) fails the clean majority target. Connectivity above the spectral boundary permits early error invasion under fixed exposure per edge. A nonempty interval between these constraints is the candidate operating region. Its location depends on signal accuracy, verification, correction, and lifetime exposure.

The communication budget changes the mechanism. A per-edge comparison increases total exposure when edges are added, so degree enters the invasion factor. A fixed sender budget removes this degree dependence from the homogeneous first-order threshold. Reporting topology without reporting message allocation leaves the causal interpretation incomplete.

The controlled model run gives a protocol-level mechanism check. Added edges increased first-generation offspring across the three graphs: 24 of 36 direct \(G_5-G_2\) task comparisons increased and none decreased. The unchanged one-third adoption fraction per exposure identifies the immediate mechanism as a larger exposed neighbourhood under fixed per-edge communication. Cumulative reach increased from the cycle to the degree-four graph in four of 12 full-cascade task blocks and then plateaued in every \(G_5-G_4\) comparison. The experiment therefore supports a propagation cost from additional exposure opportunities, while providing no evidence for a smooth density response or an interior empirical optimum.

Several limits define the present contribution. The clean-reliability curve assumes independent equal-quality signals and local majority aggregation. Real LLM outputs can be correlated and can exhibit conformity or shared training errors. The deterministic SEICS equations use a node-level closure, while finite agent networks have stochastic, path-dependent cascades. The real-model experiment uses 36 fictional tasks, one model endpoint, controlled source reliabilities, one-shot decisions, one nested graph family, and only 12 pre-registered full cascades. Exact compliance with the displayed rule makes it a controlled instruction-following test rather than an estimate of spontaneous hallucination contagion. The topology sequence changes degree, path length, clustering, and exposure multiplicity together. A broader study should estimate heterogeneous pairwise operators on naturally occurring held-out claims, compare graph families at matched edge count, and measure clean reliability and corrupted-message propagation within the same agent workflow.

Correction provides a direct extension. In \eqref{eq:ngm}, prophylactic verification reduces \(s_i^0\), pre-transmission rejection reduces the probability of reaching I, and reactive correction shortens infectious duration through \(\gamma_i\). These factors enter different sides of \(K\), which permits node-specific intervention design. Optimizing their allocation under a fixed verification budget is left for subsequent work.

\section{Conclusion}

We formulated a correction-aware network model for one traceable error in an LLM multi-agent system and coupled its spectral invasion condition to a clean-reliability requirement. The resulting graph-class theorem identifies when reliable and strictly subcritical communication graphs exist, when they are impossible, and when an intermediate connectivity range is available. A sender-normalized analysis shows that this range depends on how communication exposure scales with degree. Finite stochastic simulations support the directional mechanism. In the controlled real-model experiment, added edges increased first-generation offspring across 36 paired tasks, while cumulative erroneous adoption increased from degree 2 to degree 4 and then plateaued in the 12-task full-cascade subset. The constant adoption fraction per exposed neighbour ties the immediate increase to the number of exposure opportunities under the frozen protocol. The theory provides the separate reliability constraint needed to define an intermediate region and keeps its assumptions explicit.

\section*{Reproducibility statement}

The arXiv source archive contains the theorem calculations, unit tests, Gillespie simulator, plotting scripts, figure source tables, controlled xAI protocol, the 36 held-out tasks, and the audited v0.3 real-model summaries, manifest, response records, and attempt ledger. The numerical design is reproducible from the recorded base entropy and cell indices. The integer \texttt{seed} column in the trajectory table is an audit identifier and is not, by itself, a direct RNG reconstruction key.

\bibliographystyle{plainnat}
\bibliography{references}

\end{document}